\documentclass[a4paper,12pt]{article}

\usepackage{amsmath}
\usepackage{amssymb}
\usepackage{amsthm} 
\usepackage[T1]{fontenc}
\usepackage[bf,small,tableposition=top]{caption}
\usepackage{subfig}
\usepackage{setspace}
 \usepackage{multirow}
\usepackage{graphicx}
\usepackage{epstopdf}
\usepackage{array}
\usepackage{lscape}
\usepackage{booktabs,caption,fixltx2e}
\usepackage[flushleft]{threeparttable}
\usepackage{rotating}
\usepackage{natbib}
 \usepackage{hyperref}[sectionbib, super, numbers, round]
 \hypersetup{colorlinks=true,citecolor=blue, backref=true,linktocpage=true}
\usepackage{mathtools}

\theoremstyle{plain}
   \newtheorem{thm}{Theorem}

   \newtheorem{cor}{Corollary}

\begin{document}

\begin{center}
 \Large\bf{D-optimal joint best linear unbiased prediction of order statistics}
\end{center}
\begin{center}
\bf {\footnotesize Narayanaswamy Balakrishnan$^1$ and Ritwik Bhattacharya$^2$}
\end{center}
 \begin{center}
\textit{\scriptsize $^1$Department of Mathematics and Statistics, McMaster University, Hamilton, ON L8S 4K1, Canada}\\
 \textit{\scriptsize $^2$Department of Industrial Engineering, School of Engineering and Sciences, Tecnol\'{o}gico de Monterrey, Quer\'{e}taro 76130,  M\'{e}xico}\\
 \end{center}

\begin{abstract}
In life-testing experiments, it is often of interest to predict unobserved future failure times based on observed early failure times. A point best linear unbiased predictor (BLUP) has been developed in this context by \cite{Kaminsky_1975}. In this article, we develop joint BLUPs of two future failure times based on early failure times by minimizing the determinant of the variance-covariance matrix of the predictors. The advantage of applying joint prediction is demonstrated by using a real data set. The non-existence of joint BLUPs in certain setups is also discussed. 
\end{abstract}


\section{Introduction} \label{sec1}
\paragraph{}
The issue of predicting future unobserved failure times is of great interest in  reliability life-testing experiments. The aim in this case is to predict the unobserved failure times based on observed early failures as data.  For instance, let us assume a continuous distribution with probability density function (pdf)
\begin{equation}\label{pdf}
	\frac{1}{\sigma} f\left(\frac{x-\mu}{\sigma} \right),
\end{equation}
where $\mu$ and $\sigma$ are the location and scale parameters, respectively. Let $X_{1:n} <  \cdots < X_{n:n}$ denote the $n$ ordered observations from (\ref{pdf}).  Suppose the first $r$ order statistics, representing a Type-II right censored sample, are observed. Then, our main interest is in predicting the future $(n-r)$ unobserved failure times, that is, $(r+1)$th, $(r+2)$th, $\cdots$, $n$th failure times, based on the first $r$ order statistics observed. This can also be viewed in the context of an $n$-component parallel system. Based on the information on the first $r$ components that have already failed, we may wish to predict the system failure time, which is simply the $n$th ordered failure time.  \cite{Kaminsky_1975} obtained the best linear unbiased predictor (BLUP) of $X_{s:n}$ based on the observed values of $X_{1:n} < X_{2:n} < \cdots < X_{r:n}$, where $1\leq r<s\leq n$. Their predictor was developed  based on the result of \cite{Goldberger_1962}, who studied the best linear unbiased prediction in a generalized linear regression model. \cite{Doganaksoy_1997} showed, under Gauss-Markov model, that the best linear unbiased estimators of the model parameters remain unchanged if the predicted values of the dependent variables based on best linear unbiased predictors are treated as observed values in the best linear unbiased estimation of parameters. This property does simplify considerably the computation of point predictors. The best linear unbiased estimations, in addition to this feature, also posses some more interesting properties as shown by \cite{Rao_1997}. Due to the great importance of this prediction in life-testing experiments, \cite{Nelson_book} has presented numerous tables for the best linear unbiased prediction of order statistics for different lifetime distributions of interest. \\

 Prediction of order statistics can also be used for detecting outliers \citep[see][]{Balasooriya_1989}. An interval prediction is commonly considered instead of point prediction.  However, the interval prediction requires approximate methods for most models, with the exception of a few like exponential. A detailed survey of various prediction intervals can be found in \cite{Patel_1989}. For a comprehensive review of prediction problems in ordered data, one may refer to \cite{Kaminsky_1998}. While there are many such works on point prediction and prediction intervals for ordered data, the joint prediction of order statistics has not been dealt with. By joint prediction, we mean the simultaneous prediction of two or more order statistics. This provides the motivation for the present work. In  subsequent sections, we derive analytical expressions for the joint best linear unbiased predictors of two future order statistics by minimizing the determinant of the variance-covariance matrix of the predictors, resulting in the joint predictors being D-optimal. The advantage of joint predictors over marginal predictors is demonstrated by analyzing the design efficiency. The non-existence of joint BLUPs in certain situations is also established. \\

The rest of this paper is organized as follows. Analytic expressions of the joint predictors are derived in detail in Section 2. The non-existence of BLUPs in certain specific situations is established in Section 3. The proposed method is illustrated through some numerical results in Section 4. A real-life data is analyzed to demonstrate the advantage of joint predictors over marginal predictors. Finally, some concluding remarks are made in Section 5.

\section{Joint best linear unbiased predictors} 
\paragraph{}
In this section, we derive explicit expressions for the joint predictors of two order statistics under D-optimality criterion. For this purpose, let us denote the vector of first $r$ order statistics (observed from a life-test) from a sample of size $n$ by  $$\boldsymbol{X} = (X_{1:n}, \cdots, X_{r:n})^{\prime}_{r\times1}.$$We are then interested in the joint predictors $\tilde{X}_s$ and $\tilde{X}_t$ of $X_{s:n}$ and $X_{t:n}$, respectively, where $r<s<t\leq n$. Let $\alpha_i$ denote the expected value of the standardized order statistic $$Z_{i:n}=\frac{X_{i:n}-\mu}{\sigma}, ~i=1,\cdots,n,$$and let us further denote $$\boldsymbol{\alpha}=(\alpha_1, \cdots, \alpha_r)^{\prime}_{r\times1}.$$Also, let us denote the variance-covariance matrix of $\boldsymbol{X}$ by $\sigma^2 \boldsymbol{\Sigma}$, where $ \boldsymbol{\Sigma}$ is the $r\times r$ covariance matrix of $Z_{i:n}, i=1,\cdots, r$. In this notation, the marginal best linear unbiased predictor $\hat{X}_s$ of $X_{s:n}$ has been derived by \cite{Kaminsky_1975}, using the results of \cite{Goldberger_1962},  as
\begin{equation}\label{BLUPs}
	\hat{X}_s = \hat{\mu} + \hat{\sigma}\alpha_s + \boldsymbol{\omega}^{\prime}_s\boldsymbol{\Sigma}^{-1}(\boldsymbol{X}-\hat{\mu}\boldsymbol{1} - \hat{\sigma}\boldsymbol{\alpha}),
\end{equation}where $\boldsymbol{1} = (1,\cdots, 1)^{\prime}_{r\times 1},$ and $\boldsymbol{\omega}_s = (\omega_1,\cdots, \omega_r)^{\prime}_{r\times 1}$, with $\omega_i = \mbox{Cov}(Z_{i:n}, Z_{s:n})$.  Similarly, the marginal best linear unbiased predictor $\hat{X}_t$ of $X_{t:n}$ is exactly as in (\ref{BLUPs}) with $\alpha_s$ and $\boldsymbol{\omega}^{\prime}_s$ being replaced by $\alpha_t$ and $\boldsymbol{\omega}^{\prime}_t$, respectively. In (\ref{BLUPs}), $\hat{\mu}$ and $\hat{\sigma}$ are the best linear unbiased estimates (BLUEs) of $\mu$ and $\sigma$, respectively, based on $\boldsymbol{X} $, given by
{\small{\begin{eqnarray}\nonumber
			\hat{\mu} &=& \frac{1}{\bigtriangleup}\{ (\boldsymbol{\alpha}^{\prime} \boldsymbol{\Sigma}^{-1} \boldsymbol{\alpha})(\boldsymbol{1}^{\prime} \boldsymbol{\Sigma}^{-1}) -  (\boldsymbol{\alpha}^{\prime} \boldsymbol{\Sigma}^{-1} \boldsymbol{1})(\boldsymbol{\alpha}^{\prime} \boldsymbol{\Sigma}^{-1})\}\boldsymbol{X},\\\nonumber
			\hat{\sigma} &=& \frac{1}{\bigtriangleup}\{ (\boldsymbol{1}^{\prime} \boldsymbol{\Sigma}^{-1} \boldsymbol{1})(\boldsymbol{\alpha}^{\prime} \boldsymbol{\Sigma}^{-1}) -  (\boldsymbol{1}^{\prime} \boldsymbol{\Sigma}^{-1} \boldsymbol{\alpha})(\boldsymbol{1}^{\prime} \boldsymbol{\Sigma}^{-1})\}\boldsymbol{X},\\\nonumber
			\mbox{Var}(	\hat{\mu}) &=& \frac{\boldsymbol{\alpha}^{\prime} \boldsymbol{\Sigma}^{-1} \boldsymbol{\alpha}}{\bigtriangleup} \sigma^2,~ \mbox{Var}(	\hat{\sigma}) ~=~ \frac{\boldsymbol{1}^{\prime} \boldsymbol{\Sigma}^{-1} \boldsymbol{1}}{\bigtriangleup}\sigma^2,~ \mbox{Cov}(	\hat{\mu},	\hat{\sigma}) ~=~ -\frac{\boldsymbol{1}^{\prime} \boldsymbol{\Sigma}^{-1} \boldsymbol{\alpha}}{\bigtriangleup}\sigma^2,\\\label{bigdelta}
			\bigtriangleup &=& (\boldsymbol{1}^{\prime} \boldsymbol{\Sigma}^{-1} \boldsymbol{1}) (\boldsymbol{\alpha}^{\prime} \boldsymbol{\Sigma}^{-1} \boldsymbol{\alpha}) - (\boldsymbol{1}^{\prime} \boldsymbol{\Sigma}^{-1} \boldsymbol{\alpha})^2,
\end{eqnarray}}}with $\bigtriangleup$ being the generalized variance of BLUEs ($\hat{\mu}, \hat{\sigma}$) based on $\boldsymbol{X}$; see \cite{Cohen_book} for pertinent details.\\

\begin{thm}
	The joint best linear unbiased predictors  $\tilde{X}_s$ and $\tilde{X}_t$, determined by the D-optimality criterion, are of the form $\tilde{X}_s = \boldsymbol{a}^{\prime}\boldsymbol{X}$ and $\tilde{X}_t = \boldsymbol{b}^{\prime}\boldsymbol{X}$ in which the coefficients $\boldsymbol{a}=(a_1, \cdots, a_r)^{\prime}_{r\times 1}$ and $\boldsymbol{b}=(b_1, \cdots, b_r)^{\prime}_{r\times 1}$ are given by
	\begin{equation}\label{a}
		\boldsymbol{a} = \frac{1}{\bigtriangleup} \begin{bmatrix}
			\displaystyle\sum_{i=1}^{r} (\alpha_i - \alpha_s) (S_iR_1 - R_iS_1)\\
			\displaystyle\sum_{i=1}^{r} (\alpha_i - \alpha_s) (S_iR_2 - R_iS_2)\\
			\vdots\\
			\displaystyle\sum_{i=1}^{r} (\alpha_i - \alpha_s) (S_iR_r - R_iS_r)
		\end{bmatrix} 
	\end{equation}
and 
\begin{equation}\label{b}
	\boldsymbol{b} = \frac{1}{\bigtriangleup} \begin{bmatrix}
		\displaystyle\sum_{i=1}^{r} (\alpha_i - \alpha_t) (R_iS_1 - S_iR_1)\\
		\displaystyle\sum_{i=1}^{r} (\alpha_i - \alpha_t) (R_iS_2 - S_iR_2)\\
		\vdots\\
		\displaystyle\sum_{i=1}^{r} (\alpha_i - \alpha_t) (R_iS_r - S_iR_r)
	\end{bmatrix},
\end{equation}where $R_i$ and $S_i$ are the sums of the $i$th rows of the matrices $\boldsymbol{\Sigma}^{-1}$ and $\boldsymbol{\Sigma}^{-1}\boldsymbol{\alpha}$, respectively, and $\bigtriangleup$ is as defined earlier in (\ref{bigdelta}).  
\end{thm}
\begin{proof}

 The BLUPs will now be derived jointly by minimizing the determinant of the variance-covariance matrix of BLUPs with respect to vectors $ \boldsymbol{a}$ and $ \boldsymbol{b}$. Note that the variance-covariance matrix, say $\boldsymbol{V}$, of BLUPs $\tilde{X}_s$ and $\tilde{X}_t$ is of the form
\begin{equation*}
	\boldsymbol{V} = \sigma^2 \begin{bmatrix} \boldsymbol{a}^{\prime}\boldsymbol{\Sigma}\boldsymbol{a} &  \boldsymbol{a}^{\prime}\boldsymbol{\Sigma}\boldsymbol{b}\\ \boldsymbol{a}^{\prime}\boldsymbol{\Sigma}\boldsymbol{b} & \boldsymbol{b}^{\prime}\boldsymbol{\Sigma}\boldsymbol{b}\end{bmatrix}  = \sigma^2 \begin{bmatrix} V_{11} &  V_{12}\\ V_{12} & V_{22}\end{bmatrix} (\mbox{say}),
\end{equation*}with $|\boldsymbol{V}| = V_{11}V_{22}-V^2_{12}$ and $$\boldsymbol{V}^{-1} = \frac{1}{|\boldsymbol{V}|}\begin{bmatrix} V_{22} &  -V_{12}\\ -V_{12} & V_{11}\end{bmatrix}.$$Note that we can leave the multiplicative factor $\sigma^2$ in the minimization process. Moreover, while minimizing $|\boldsymbol{V}|$, we must impose four constraints due to the unbiasedness of BLUPs and those are evidently $\boldsymbol{a}^{\prime}\boldsymbol{1}=1$, $\boldsymbol{a}^{\prime}\boldsymbol{\alpha}=\alpha_s$, $\boldsymbol{b}^{\prime}\boldsymbol{1}=1$ and $\boldsymbol{b}^{\prime}\boldsymbol{\alpha}=\alpha_t$. So, the Lagrangian method can be employed to determine the optimal $\boldsymbol{a}$ and $\boldsymbol{b}$ by considering the objective function
\begin{eqnarray}\nonumber
Q(\boldsymbol{a}, \boldsymbol{b}) &=& (\boldsymbol{a}^{\prime}\boldsymbol{\Sigma}\boldsymbol{a})(\boldsymbol{b}^{\prime}\boldsymbol{\Sigma}\boldsymbol{b})	- (\boldsymbol{a}^{\prime}\boldsymbol{\Sigma}\boldsymbol{b})^2 - 2\lambda_1(\boldsymbol{a}^{\prime}\boldsymbol{1}-1) - 2\lambda^{*}_1(\boldsymbol{a}^{\prime}\boldsymbol{\alpha}-\alpha_s)\\\label{Lobg}
&&-2\lambda_2(\boldsymbol{b}^{\prime}\boldsymbol{1}-1) - 2\lambda^{*}_2(\boldsymbol{b}^{\prime}\boldsymbol{\alpha}-\alpha_t).
\end{eqnarray}
Differentiating  (\ref{Lobg}) with respect to $ \boldsymbol{a}$ and $ \boldsymbol{b}$ and equating them to vector $ \boldsymbol{0}$ of dimension $r\times 1$, we obtain 
\begin{eqnarray}\label{E1}
	(\boldsymbol{b}^{\prime}\boldsymbol{\Sigma}\boldsymbol{b})\boldsymbol{\Sigma}\boldsymbol{a} - (\boldsymbol{a}^{\prime}\boldsymbol{\Sigma}\boldsymbol{b})\boldsymbol{\Sigma}\boldsymbol{b} - \lambda_1\boldsymbol{1} - \lambda^{*}_1\boldsymbol{\alpha} = \boldsymbol{0},\\\label{E3}
	(\boldsymbol{a}^{\prime}\boldsymbol{\Sigma}\boldsymbol{a})\boldsymbol{\Sigma}\boldsymbol{b} - (\boldsymbol{a}^{\prime}\boldsymbol{\Sigma}\boldsymbol{b})\boldsymbol{\Sigma}\boldsymbol{a} - \lambda_2\boldsymbol{1} - \lambda^{*}_2\boldsymbol{\alpha} = \boldsymbol{0},
\end{eqnarray}respectively. Pre-multiplying (\ref{E1}) by $\boldsymbol{a}^{\prime}$ and simplifying the resulting equation, we obtain
\begin{equation}\label{E2}
	\lambda_1 + \lambda^{*}_1\alpha_s = |\boldsymbol{V}|.
\end{equation}Similarly, pre-multiplying (\ref{E3}) by $\boldsymbol{b}^{\prime}$  and simplifying the resulting equation, we obtain
\begin{equation}\label{E4}
	\lambda_2 + \lambda^{*}_2\alpha_t = |\boldsymbol{V}|.
\end{equation}
Next, pre-multiplying (\ref{E1}) by $\boldsymbol{b}^{\prime}$ and simplifying the resulting equation, we obtain
\begin{equation}\label{E5}
	\lambda_1 + \lambda^{*}_1\alpha_t = 0.
\end{equation}
Finally, pre-multiplying (\ref{E3}) by $\boldsymbol{a}^{\prime}$ and simplifying the resulting equation, we obtain
\begin{equation}\label{E6}
	\lambda_2 + \lambda^{*}_2\alpha_s = 0.
\end{equation}
By solving (\ref{E2}) and (\ref{E5}), we obtain
\begin{equation*}
	\lambda_1 = -\frac{|\boldsymbol{V}|\alpha_t}{\alpha_s - \alpha_t},~
\lambda^{*}_1 = \frac{|\boldsymbol{V}|}{\alpha_s - \alpha_t},
\end{equation*}and similarly by solving (\ref{E4}) and (\ref{E6}), we obtain
\begin{equation*}
\lambda_2 = -\frac{|\boldsymbol{V}|\alpha_s}{\alpha_t - \alpha_s},~
\lambda^{*}_2 = \frac{|\boldsymbol{V}|}{\alpha_t - \alpha_s}.
\end{equation*}
Next, pre-multiplying (\ref{E1}) by $\boldsymbol{\Sigma}^{-1}$ and substituting for $\lambda_1$ and $\lambda^{*}_1$, we obtain
\begin{equation}\label{E11}
	V_{22} \boldsymbol{a} -  V_{12} \boldsymbol{b} = \frac{|\boldsymbol{V}|}{\alpha_s - \alpha_t} \boldsymbol{\Sigma}^{-1} (\boldsymbol{\alpha} - \alpha_t\boldsymbol{1}).
\end{equation}Similarly,  by pre-multiplying (\ref{E3}) by $\boldsymbol{\Sigma}^{-1}$ and substituting the values of $\lambda_2$ and $\lambda^{*}_2$, we obtain
\begin{equation}\label{E12}
	V_{11} \boldsymbol{b} -  V_{12} \boldsymbol{a} = \frac{|\boldsymbol{V}|}{\alpha_t - \alpha_s} \boldsymbol{\Sigma}^{-1} (\boldsymbol{\alpha} - \alpha_s\boldsymbol{1}).
\end{equation}Note that $\boldsymbol{\alpha} - \alpha_s\boldsymbol{1} = (\alpha_1 - \alpha_s,\cdots,\alpha_r - \alpha_s)^{\prime}_{r\times 1}$ and $\boldsymbol{\alpha} - \alpha_t\boldsymbol{1} = (\alpha_1 - \alpha_t, \cdots,\alpha_r - \alpha_t)^{\prime}_{r\times 1}$, and let us denote them by $\boldsymbol{\alpha}_{(s)}$ and $\boldsymbol{\alpha}_{(t)}$, respectively. Thence, by writing the Eqs. (\ref{E11}) and (\ref{E12}) in a matrix form as
\begin{equation*}
	\begin{bmatrix}
		\boldsymbol{a} & \boldsymbol{b}
	\end{bmatrix}\begin{bmatrix}
	V_{22} & -V_{12}\\
	-V_{12} & V_{11}
\end{bmatrix} = \frac{|\boldsymbol{V}|}{\alpha_t - \alpha_s}\boldsymbol{\Sigma}^{-1}\begin{bmatrix}
-\boldsymbol{\alpha}_{(t)} & \boldsymbol{\alpha}_{(s)}
\end{bmatrix},
\end{equation*}we find 
\begin{equation}\label{E14}
	\begin{bmatrix}
		\boldsymbol{a} & \boldsymbol{b}
	\end{bmatrix} = \frac{1}{\alpha_t - \alpha_s}\boldsymbol{\Sigma}^{-1}\begin{bmatrix}
	-\boldsymbol{\alpha}_{(t)} & \boldsymbol{\alpha}_{(s)}
\end{bmatrix}\boldsymbol{V}.
\end{equation} 
The solution in (\ref{E14}) explicitly gives
\begin{eqnarray}\label{E15}
\boldsymbol{a} &=&  \frac{1}{\alpha_t - \alpha_s} \{  -V_{11}\boldsymbol{\Sigma}^{-1}\boldsymbol{\alpha}_{(t)} + V_{12}\boldsymbol{\Sigma}^{-1}\boldsymbol{\alpha}_{(s)}    \}, \\\label{E16}
\boldsymbol{b} &=&  \frac{1}{\alpha_t - \alpha_s} \{  -V_{12}\boldsymbol{\Sigma}^{-1}\boldsymbol{\alpha}_{(t)} + V_{22}\boldsymbol{\Sigma}^{-1}\boldsymbol{\alpha}_{(s)}    \}.
\end{eqnarray}	
Now, the unbiasedness conditions $\boldsymbol{a}^{\prime}\boldsymbol{1}=1$ and $\boldsymbol{a}^{\prime}\boldsymbol{\alpha}=\alpha_s$ give
\begin{eqnarray}\label{E17}
	\frac{1}{\alpha_t - \alpha_s} \{  -V_{11}(\boldsymbol{\alpha}_{(t)}^{\prime}\boldsymbol{\Sigma}^{-1}\boldsymbol{1}) + V_{12}(\boldsymbol{\alpha}_{(s)}^{\prime}\boldsymbol{\Sigma}^{-1}\boldsymbol{1})    \} &=& 1,\\\label{E18}
\frac{1}{\alpha_t - \alpha_s} \{  -V_{11}(\boldsymbol{\alpha}_{(t)}^{\prime}\boldsymbol{\Sigma}^{-1}\boldsymbol{\alpha}) + V_{12}(\boldsymbol{\alpha}_{(s)}^{\prime}\boldsymbol{\Sigma}^{-1}\boldsymbol{\alpha})    \} &=& \alpha_s.
\end{eqnarray}Solving (\ref{E17}) and (\ref{E18}) for $V_{11}$ and $V_{12}$, we get
\begin{eqnarray} \label{E19}
	V_{11}	&=& -(\alpha_t - \alpha_s)\frac{\boldsymbol{\alpha}_{(s)}^{\prime}\boldsymbol{\Sigma}^{-1}\boldsymbol{\alpha}_{(s)}}{	\left(\boldsymbol{\alpha}_{(t)}^{\prime}\boldsymbol{\Sigma}^{-1}\boldsymbol{1}\boldsymbol{\alpha}_{(s)}^{\prime} - \boldsymbol{\alpha}_{(s)}^{\prime}\boldsymbol{\Sigma}^{-1}\boldsymbol{1}\boldsymbol{\alpha}_{(t)}^{\prime}\right)\boldsymbol{\Sigma}^{-1}\boldsymbol{\alpha}},\\\label{E20}
	V_{12} &=& -(\alpha_t - \alpha_s)\frac{\boldsymbol{\alpha}_{(t)}^{\prime}\boldsymbol{\Sigma}^{-1}\boldsymbol{\alpha}_{(s)}}{ \left(\boldsymbol{\alpha}_{(t)}^{\prime}\boldsymbol{\Sigma}^{-1}\boldsymbol{1}\boldsymbol{\alpha}_{(s)}^{\prime} - \boldsymbol{\alpha}_{(s)}^{\prime}\boldsymbol{\Sigma}^{-1}\boldsymbol{1}\boldsymbol{\alpha}_{(t)}^{\prime}\right)\boldsymbol{\Sigma}^{-1}\boldsymbol{\alpha}}.
\end{eqnarray}Now, we shall further simplify the above derived expressions. For this let us denote 
\begin{equation*}
	\boldsymbol{\Sigma} = \begin{bmatrix}
		\sigma_{11} & \sigma_{12}&\cdots& \sigma_{1r}\\
		\sigma_{12} & \sigma_{22}&\cdots& \sigma_{2r}\\
		\vdots & \vdots& \cdots & \vdots \\
		\sigma_{1r} & \sigma_{2r} &\cdots & \sigma_{rr}
	\end{bmatrix}\mbox{and}~
\boldsymbol{\Sigma}^{-1} = \begin{bmatrix}
	\sigma^{11} & \sigma^{12}&\cdots& \sigma^{1r}\\
	\sigma^{12} & \sigma^{22}&\cdots& \sigma^{2r}\\
	\vdots & \vdots& \cdots & \vdots \\
	\sigma^{1r} & \sigma^{2r} &\cdots & \sigma^{rr}
\end{bmatrix}.
\end{equation*}Then,
\begin{equation*}
	\boldsymbol{\Sigma}^{-1}\boldsymbol{1}= \begin{bmatrix}
			\sigma^{11} + \sigma^{12}+\cdots+ \sigma^{1r}\\
		\sigma^{12} + \sigma^{22}+\cdots+ \sigma^{2r}\\
		 \vdots \\
		\sigma^{1r} + \sigma^{2r} +\cdots + \sigma^{rr}
	\end{bmatrix}=\begin{bmatrix}
	R_1\\
	R_2\\
	\vdots \\
	R_r
\end{bmatrix},
\end{equation*}where $R_i = \displaystyle\sum_{j=1}^{r} \sigma^{ji}$ is the sum of the $i$th row of $\boldsymbol{\Sigma}^{-1}$, for $i=1, \cdots, r$, and similarly
\begin{equation*}
	\boldsymbol{\Sigma}^{-1}\boldsymbol{\alpha}= \begin{bmatrix}
		\alpha_1\sigma^{11} + \alpha_2\sigma^{12}+\cdots+ \alpha_r\sigma^{1r}\\
		\alpha_1\sigma^{12} + \alpha_2\sigma^{22}+\cdots+ \alpha_r\sigma^{2r}\\
		\vdots \\
		\alpha_1\sigma^{1r} + \alpha_2\sigma^{2r} +\cdots + \alpha_r\sigma^{rr}
	\end{bmatrix}=\begin{bmatrix}
		S_1\\
		S_2\\
		\vdots \\
		S_r
	\end{bmatrix},
\end{equation*}where $S_i = \displaystyle\sum_{j=1}^{r} \alpha_j\sigma^{ji},$ for $i=1, \cdots, r$, is the sum of $i$th row of $	\boldsymbol{\Sigma}^{-1}\boldsymbol{\alpha}$. With these notations, we have
\begin{eqnarray*}
	\boldsymbol{\alpha}_{(t)}^{\prime}\boldsymbol{\Sigma}^{-1}\boldsymbol{1}\boldsymbol{\alpha}_{(s)}^{\prime} &=&  \left(\sum_{i=1}^{r} R_i(\alpha_i - \alpha_t)\right)\left[\alpha_1 - \alpha_s, \cdots, \alpha_r - \alpha_s\right],\\
 \boldsymbol{\alpha}_{(s)}^{\prime}\boldsymbol{\Sigma}^{-1}\boldsymbol{1}\boldsymbol{\alpha}_{(t)}^{\prime} &=&  \left(\sum_{i=1}^{r} R_i(\alpha_i - \alpha_s)\right)\left[\alpha_1 - \alpha_t, \cdots, \alpha_r - \alpha_t\right],	
\end{eqnarray*}and consequently,	
\begin{equation*}
	\boldsymbol{\alpha}_{(t)}^{\prime}\boldsymbol{\Sigma}^{-1}\boldsymbol{1}\boldsymbol{\alpha}_{(s)}^{\prime} -   \boldsymbol{\alpha}_{(s)}^{\prime}\boldsymbol{\Sigma}^{-1}\boldsymbol{1}\boldsymbol{\alpha}_{(t)}^{\prime} = (\alpha_s - \alpha_t)\left[\alpha_1R - R^{*},\cdots, \alpha_rR - R^{*}\right],
\end{equation*}where $R=\sum_{i=1}^r R_i$ and $R^{*}=\sum_{i=1}^r R_i\alpha_i$. Observe that $R =  \sum_{i=1}^r R_i ~=~ \boldsymbol{1}^{\prime} \boldsymbol{\Sigma}^{-1} \boldsymbol{1}$, $R^{*} = \sum_{i=1}^r R_i\alpha_i ~=~ \boldsymbol{\alpha}^{\prime} \boldsymbol{\Sigma}^{-1} \boldsymbol{1}$, $\sum_{i=1}^r S_i = \boldsymbol{1}^{\prime} \boldsymbol{\Sigma}^{-1} \boldsymbol{\alpha} ~=~ \boldsymbol{\alpha}^{\prime} \boldsymbol{\Sigma}^{-1} \boldsymbol{1}$ and $\sum_{i=1}^r \alpha_iS_i = \boldsymbol{\alpha}^{\prime} \boldsymbol{\Sigma}^{-1} \boldsymbol{\alpha}$. With these expressions, the denominators of (\ref{E19}) and (\ref{E20}) become
\begin{eqnarray*}
&&(\alpha_s - \alpha_t)\left[\alpha_1R - R^{*}, \alpha_2R - R^{*},\cdots, \alpha_rR - R^{*}\right]\begin{bmatrix}
		S_1\\
		S_2\\
		\vdots \\
		S_r
	\end{bmatrix}\\
&=& (\alpha_s - \alpha_t) \left( R\sum_{i=1}^r \alpha_iS_i - R^{*}\sum_{i=1}^r S_i \right)\\
&=& (\alpha_s - \alpha_t) \bigtriangleup.
\end{eqnarray*}So, from (\ref{E19}) and (\ref{E20}) we simply obtain
\begin{equation*}
	V_{11} = \frac{\boldsymbol{\alpha}_{(s)}^{\prime}\boldsymbol{\Sigma}^{-1}\boldsymbol{\alpha}_{(s)}}{\bigtriangleup}\mbox{~and~} V_{12} = \frac{\boldsymbol{\alpha}_{(t)}^{\prime}\boldsymbol{\Sigma}^{-1}\boldsymbol{\alpha}_{(s)}}{\bigtriangleup}.
\end{equation*}

If we use the unbiasedness conditions $\boldsymbol{b}^{\prime}\boldsymbol{1}=1$ and $\boldsymbol{b}^{\prime}\boldsymbol{\alpha}=\alpha_t$ and proceed exactly as above, we obtain
\begin{equation*}
	V_{12} = \frac{\boldsymbol{\alpha}_{(t)}^{\prime}\boldsymbol{\Sigma}^{-1}\boldsymbol{\alpha}_{(s)}}{\bigtriangleup}\mbox{~and~} V_{22} = \frac{\boldsymbol{\alpha}_{(t)}^{\prime}\boldsymbol{\Sigma}^{-1}\boldsymbol{\alpha}_{(t)}}{\bigtriangleup},
\end{equation*}showing the uniqueness of the obtained solution.	Substituting for $V_{11}$ and $V_{12}$ in (\ref{E15}), we obtain
\begin{equation}\label{E32}
	\boldsymbol{a} = \frac{1}{(\alpha_t - \alpha_s)\bigtriangleup} \{ (\boldsymbol{\alpha}_{(s)}^{\prime}\boldsymbol{\Sigma}^{-1}\boldsymbol{\alpha}_{(t)})  \boldsymbol{\Sigma}^{-1}\boldsymbol{\alpha}_{(s)} - (\boldsymbol{\alpha}_{(s)}^{\prime}\boldsymbol{\Sigma}^{-1}\boldsymbol{\alpha}_{(s)})  \boldsymbol{\Sigma}^{-1}\boldsymbol{\alpha}_{(t)} \},
\end{equation}which, after some algebraic calculations, can be shown to be as given in (\ref{a}). Similarly, substituting $V_{12}$ and $V_{22}$ in (\ref{E16}) and simplifying, we obtain an explicit expression for the coefficient vector $\boldsymbol{b}$ as given in (\ref{b}).
\end{proof}
\begin{cor}
In the process of the above derivation, the variance-covariance matrix of the joint predictors is readily found to be
{\small{\begin{equation*}
	\mbox{Var}\begin{pmatrix}
		\tilde{X}_s\\
		\tilde{X}_t
	\end{pmatrix}= \mbox{Var}\begin{pmatrix}
		\boldsymbol{a}^{\prime} \boldsymbol{X} \\
		\boldsymbol{b}^{\prime} \boldsymbol{X} 
	\end{pmatrix} = \sigma^2 \begin{bmatrix}
	\boldsymbol{a}^{\prime} \boldsymbol{\Sigma}\boldsymbol{a} & \boldsymbol{a}^{\prime} \boldsymbol{\Sigma}\boldsymbol{b}\\
	\boldsymbol{a}^{\prime} \boldsymbol{\Sigma}\boldsymbol{b} & \boldsymbol{b}^{\prime} \boldsymbol{\Sigma}\boldsymbol{b}
\end{bmatrix} = \frac{\sigma^2}{\bigtriangleup}\begin{bmatrix}
\boldsymbol{\alpha}_{(s)}^{\prime}\boldsymbol{\Sigma}^{-1}\boldsymbol{\alpha}_{(s)} & \boldsymbol{\alpha}_{(t)}^{\prime}\boldsymbol{\Sigma}^{-1}\boldsymbol{\alpha}_{(s)}\\
\boldsymbol{\alpha}_{(t)}^{\prime}\boldsymbol{\Sigma}^{-1}\boldsymbol{\alpha}_{(s)} & \boldsymbol{\alpha}_{(t)}^{\prime}\boldsymbol{\Sigma}^{-1}\boldsymbol{\alpha}_{(t)}
\end{bmatrix},
\end{equation*}}}where $\bigtriangleup$ is the generalized variance as given in (\ref{bigdelta}). 
\end{cor}

\section{Non-existence of joint prediction}
\paragraph{}
In this section, we discuss two cases in which the joint predictors do not exist. 
\subsection{Prediction of more than two order statistics in location-scale family}
\paragraph{}
We now formally show that the joint prediction of more than two order statistics is not possible when the lifetimes come from a location-scale family of distributions. We will first prove our claim for the joint prediction of three order statistics. 

Suppose we wish to jointly predict three order statistics $X_{s:n}$, $X_{t:n}$ and $X_{u:n}$ for $r<s<t<u\leq n$. Now, let us assume that the joint predictors of them are of the form $\tilde{X}_s = \boldsymbol{a}^{\prime}\boldsymbol{X}$, $\tilde{X}_t = \boldsymbol{b}^{\prime}\boldsymbol{X}$ and $\tilde{X}_u = \boldsymbol{c}^{\prime}\boldsymbol{X}$, respectively, where $\boldsymbol{a}^{\prime}$, $\boldsymbol{b}^{\prime}$ and $\boldsymbol{c}^{\prime}$ are all coefficient vectors of dimension $r\times 1$. By applying the Lagrangian method as done in Section 2, we will obtain 
\begin{eqnarray*}
	\boldsymbol{a}^{\prime} &=& 	\frac{V_{11}}{\alpha_t - \alpha_s}(\alpha_t\boldsymbol{1}^{\prime}\boldsymbol{\Sigma}^{-1} -\boldsymbol{\alpha}^{\prime}\boldsymbol{\Sigma}^{-1}) + 
	\frac{V_{12}}{\alpha_u - \alpha_t}(\alpha_u\boldsymbol{1}^{\prime}\boldsymbol{\Sigma}^{-1} -\boldsymbol{\alpha}^{\prime}\boldsymbol{\Sigma}^{-1})\\
	&& 	+\frac{V_{13}}{\alpha_s - \alpha_u}(\alpha_s\boldsymbol{1}^{\prime}\boldsymbol{\Sigma}^{-1} -\boldsymbol{\alpha}^{\prime}\boldsymbol{\Sigma}^{-1}),\\
		\boldsymbol{b}^{\prime} &=& 	\frac{V_{12}}{\alpha_t - \alpha_s}(\alpha_t\boldsymbol{1}^{\prime}\boldsymbol{\Sigma}^{-1} -\boldsymbol{\alpha}^{\prime}\boldsymbol{\Sigma}^{-1}) + 
		\frac{V_{22}}{\alpha_u - \alpha_t}(\alpha_u\boldsymbol{1}^{\prime}\boldsymbol{\Sigma}^{-1} -\boldsymbol{\alpha}^{\prime}\boldsymbol{\Sigma}^{-1})\\
		&&+\frac{V_{23}}{\alpha_s - \alpha_u}(\alpha_s\boldsymbol{1}^{\prime}\boldsymbol{\Sigma}^{-1} -\boldsymbol{\alpha}^{\prime}\boldsymbol{\Sigma}^{-1}),\\
	\end{eqnarray*}
\begin{eqnarray*}
		\boldsymbol{c}^{\prime}&=& 	\frac{V_{13}}{\alpha_t - \alpha_s}(\alpha_t\boldsymbol{1}^{\prime}\boldsymbol{\Sigma}^{-1} -\boldsymbol{\alpha}^{\prime}\boldsymbol{\Sigma}^{-1}) + 
		\frac{V_{23}}{\alpha_u - \alpha_t}(\alpha_u\boldsymbol{1}^{\prime}\boldsymbol{\Sigma}^{-1} -\boldsymbol{\alpha}^{\prime}\boldsymbol{\Sigma}^{-1})\\
		&&+\frac{V_{33}}{\alpha_s - \alpha_u}(\alpha_s\boldsymbol{1}^{\prime}\boldsymbol{\Sigma}^{-1} -\boldsymbol{\alpha}^{\prime}\boldsymbol{\Sigma}^{-1}),
\end{eqnarray*}where the quantities $V_{11}$, $V_{12}$, $V_{13}$, $V_{22}$, $V_{23}$ and $V_{33}$, all of which are to be determined, are the elements of the variance-covariance matrix of the joint predictors given by
\begin{equation*}
	\sigma^2\begin{bmatrix}
		V_{11} & V_{12} & V_{13} \\
		V_{12} & V_{22} & V_{23} \\
		V_{13} & V_{23} & V_{33}
	\end{bmatrix} =
\sigma^2\begin{bmatrix}
	\boldsymbol{a}^{\prime}\boldsymbol{\Sigma}	\boldsymbol{a} & 	\boldsymbol{a}^{\prime}\boldsymbol{\Sigma}	\boldsymbol{b} & 	\boldsymbol{a}^{\prime}\boldsymbol{\Sigma}	\boldsymbol{c}\\
		\boldsymbol{a}^{\prime}\boldsymbol{\Sigma}	\boldsymbol{b} & 	\boldsymbol{b}^{\prime}\boldsymbol{\Sigma}	\boldsymbol{b} & 	\boldsymbol{b}^{\prime}\boldsymbol{\Sigma}	\boldsymbol{c}\\
			\boldsymbol{a}^{\prime}\boldsymbol{\Sigma}	\boldsymbol{c} & 	\boldsymbol{b}^{\prime}\boldsymbol{\Sigma}	\boldsymbol{c} & 	\boldsymbol{c}^{\prime}\boldsymbol{\Sigma}	\boldsymbol{c}
\end{bmatrix}.
\end{equation*}Now, observe that the two unbiasedness conditions on $\tilde{X}_{s}$, namely, $\boldsymbol{a}^{\prime}\boldsymbol{1} = 1$ and $\boldsymbol{a}^{\prime}\boldsymbol{\alpha} = \alpha_s$ will generate two linear equations in three unknowns $V_{11}$,  $V_{12}$ and $V_{13}$. The other unbiasedness conditions will similarly yield four other linear equations, and thus we end up with a system of the form  
\begin{equation}\label{LE}
	\begin{bmatrix}
		A_{11} & A_{12} & A_{13} & 0 & 0 & 0 \\
		B_{11} & B_{12} & B_{13} & 0 & 0 & 0 \\ 
		0 & A_{11} & 0 & A_{12} & A_{13} & 0 \\ 
		0 & B_{11} & 0 & B_{12} & B_{13} & 0  \\
		0 & 0 & A_{11} & 0 & A_{12} & A_{13} \\
		0 & 0 & 	B_{11} & 0 & 	B_{12} & 	B_{13} &  \\ 
	\end{bmatrix}
\begin{bmatrix}
V_{11}\\
V_{12}\\
V_{13}\\
V_{22}\\
V_{23}\\
V_{33}
\end{bmatrix}= \begin{bmatrix}
1\\
\alpha_s\\
1\\
\alpha_t\\
1\\
\alpha_u
\end{bmatrix},
\end{equation}where
\begin{equation*}
	A_{11} = \frac{\alpha_tR-R^{*}}{\alpha_t - \alpha_s}, A_{12} = \frac{\alpha_uR-R^{*}}{\alpha_u - \alpha_t}, A_{13} = \frac{\alpha_sR-R^{*}}{\alpha_s - \alpha_u},
\end{equation*}
\begin{equation*}
		B_{11} = \frac{\alpha_tR^{*}-\sum_{i=1}^{r}\alpha_iS_i}{\alpha_t - \alpha_s},
	B_{12}= \frac{\alpha_uR^{*}-\sum_{i=1}^{r}\alpha_iS_i}{\alpha_u - \alpha_t},
	B_{13}= \frac{\alpha_sR^{*}-\sum_{i=1}^{r}\alpha_iS_i}{\alpha_s - \alpha_u}.
\end{equation*}
Then, a straightforward matrix decomposition shows that the determinant of the coefficient matrix in (\ref{LE}) is zero implying the non-existence of the solution in (\ref{LE}). Hence, the claim. The non-existence of joint BLUPs for more than three order statistics can be established in a similar way. 

\subsection{Joint prediction in scale family}
\paragraph{}
The non-existence of joint prediction of order statistics can also be established when the parent distribution belongs to the scale family. Assuming the joint predictors $\tilde{X}_s = \boldsymbol{a}^{\prime}\boldsymbol{X}$ and $\tilde{X}_t = \boldsymbol{b}^{\prime}\boldsymbol{X}$, we have two unbiasedness conditions, namely, $\boldsymbol{a}^{\prime}\boldsymbol{\alpha} = \alpha_s$ and $\boldsymbol{b}^{\prime}\boldsymbol{\beta} = \alpha_t$. Now, by applying the Lagrangian method with Lagrangian multipliers $\lambda_1$ and $\lambda_2$ for the two unbiasedness conditions $\boldsymbol{a}^{\prime}\boldsymbol{\alpha} = \alpha_s$ and $\boldsymbol{b}^{\prime}\boldsymbol{\beta} = \alpha_t$, respectively, we have two equations
\begin{eqnarray}\label{E30}
	\lambda_1\boldsymbol{\alpha} &=& (\boldsymbol{b}^{\prime}\boldsymbol{\Sigma}\boldsymbol{b})\boldsymbol{\Sigma}\boldsymbol{a} - (\boldsymbol{a}^{\prime}\boldsymbol{\Sigma}\boldsymbol{b})\boldsymbol{\Sigma}\boldsymbol{b},\\\label{E31}
	\lambda_2\boldsymbol{\alpha}&=& (\boldsymbol{a}^{\prime}\boldsymbol{\Sigma}\boldsymbol{a})\boldsymbol{\Sigma}\boldsymbol{b} - (\boldsymbol{a}^{\prime}\boldsymbol{\Sigma}\boldsymbol{b})\boldsymbol{\Sigma}\boldsymbol{a}.
\end{eqnarray}Upon pre-multiplying (\ref{E30}) by $\boldsymbol{b}^{\prime}$, we get
\begin{eqnarray*}
	\lambda_1\alpha_t &=& (\boldsymbol{b}^{\prime}\boldsymbol{\Sigma}\boldsymbol{b})(\boldsymbol{b}^{\prime}\boldsymbol{\Sigma}\boldsymbol{a}) - (\boldsymbol{a}^{\prime}\boldsymbol{\Sigma}\boldsymbol{b})(\boldsymbol{b}^{\prime}\boldsymbol{\Sigma}\boldsymbol{b})\\
	&=&0,
\end{eqnarray*}which implies $\lambda_1=0$. Similarly, pre-multiplying (\ref{E31}) by $\boldsymbol{a}^{\prime}$ and simplifying, we get $\lambda_2=0$. We can then readily see the non-existence of the joint predictors in this case as well.

\section{Numerical results}
\paragraph{}
For illustrative purpose, let us consider a data set presented by \cite{Schmee_1979}. The data, assumed to follow a normal distribution, have the first $r=5$ failure times to be 87.0, 92.8, 117.1, 133.6 and 138.6, and so $\boldsymbol{X}=(87.0, 92.8, 117.1, 133.6, 138.6)^{\prime}$. Table 1 then presents a comparison between the marginal predictors in (\ref{BLUPs}) and the joint predictors developed in the last section. Two forms of efficiency measures are defined as follows:
\begin{equation*}
\mbox{D-efficiency} = 	\frac{
	\begin{aligned}
		\mbox{Determinant of the variance-covariance} \\
		\mbox{matrix of joint predictors $\tilde{X}_s$ and $\tilde{X}_t$}
	\end{aligned}
}{\begin{aligned}
	\mbox{Determinant of the variance-covariance} \\
	\mbox{matrix of marginal predictors $\hat{X}_s$ and $\hat{X}_t$}
\end{aligned}},
\end{equation*}
and
\begin{equation*}
	\mbox{Trace-efficiency} = 	\frac{
		\begin{aligned}
			\mbox{Trace of the variance-covariance} \\
			\mbox{matrix of joint predictors $\tilde{X}_s$ and $\tilde{X}_t$}
		\end{aligned}
	}{\begin{aligned}
			\mbox{Trace of the variance-covariance} \\
			\mbox{matrix of marginal predictors $\hat{X}_s$ and $\hat{X}_t$}
	\end{aligned}}.
\end{equation*}Then, the gain and loss in efficiency for the joint prediction can be defined as
\begin{equation*}
	\mbox{Efficiency gain} = 1 - \mbox{D-efficiency}
\end{equation*} 
and
\begin{equation*}
	\mbox{Efficiency loss} = 1 - \mbox{Trace-efficiency},
\end{equation*}respectively. Thence, the overall gain in efficiency can be defined as
\begin{equation*}
	\mbox{Overall efficiency gain} = 	\mbox{Efficiency gain} - \mbox{Efficiency loss}.
\end{equation*}Table 1 shows that, in all cases considered, Overall efficiency gain are positive indicating the advantage of using joint prediction. Further, the coefficients $\boldsymbol{a}$ and $\boldsymbol{b}$ are presented in Table 2 for three different scenarios: (i) $s=r+1, t=r+2$; (ii) $s=r+1, t=n$; and (iii) $s=n-1, t=n$, for various choices of $n=(10, 15, 20)$ and $r=(5, 10)$.

\section{Concluding remarks}
\paragraph{}
In this work, we have developed the joint best linear unbiased predictors of two  unobserved order statistics based on observed order statistics through D-optimality criterion. The advantage of using point predictors over marginal predictors is demonstrated by a real data set. A possible future research problem would be to develop an optimal compound design for joint prediction based on both Trace-efficiency and D-efficiency introduced in the last section.

\bibliographystyle{apalike}
\bibliography{ritwik_ref}

\begin{thebibliography}{}

\bibitem[Balakrishnan and Cohen, 1991]{Cohen_book}
Balakrishnan, N. and Cohen, A.~C. (1991).
\newblock {\em Order Statistics and Inference: Estimation Methods}.
\newblock Academic Press, Boston.

\bibitem[Balakrishnan and Rao, 1997]{Rao_1997}
Balakrishnan, N. and Rao, C.~R. (1997).
\newblock A note on the best linear unbiased estimation based on order
  statistics.
\newblock {\em The American Statistician}, {\bf{51}}:181--185.

\bibitem[Balasooriya, 1989]{Balasooriya_1989}
Balasooriya, U. (1989).
\newblock Detection of outliers in the exponential distribution based on
  prediction.
\newblock {\em Communications in Statistics - Theory and Methods},
  {\bf{18}}:711--720.

\bibitem[Doganaksoy and Balakrishnan, 1997]{Doganaksoy_1997}
Doganaksoy, N. and Balakrishnan, N. (1997).
\newblock A useful property of best linear unbiased predictors with
  applications to life-testing.
\newblock {\em The American Statistician}, {\bf{51}}:22--28.

\bibitem[Goldberger, 1962]{Goldberger_1962}
Goldberger, A.~S. (1962).
\newblock Best linear unbiased prediction in the generalized linear regression
  model.
\newblock {\em Journal of the American Statistical Association},
  {\bf{57}}:369--375.

\bibitem[Kaminsky and Nelson, 1975]{Kaminsky_1975}
Kaminsky, K.~S. and Nelson, P.~I. (1975).
\newblock Best linear unbiased prediction of order statistics in location and
  scale families.
\newblock {\em Journal of the American Statistical Association},
  {\bf{70}}:145--150.

\bibitem[Kaminsky and Nelson, 1998]{Kaminsky_1998}
Kaminsky, K.~S. and Nelson, P.~I. (1998).
\newblock Prediction of order statistics.
\newblock In Balakrishnan, N. and Rao, C.~R., editors, {\em Handbook of
  Statistics, Vol. 17-Order Statistics: Applications}, pages 431--450.
  North-Holland, Amsterdam.

\bibitem[Nelson, 2003]{Nelson_book}
Nelson, W. (2003).
\newblock {\em Applied Life Data Analysis}.
\newblock John Wiley \& Sons, Hoboken, New Jersey.

\bibitem[Patel, 1989]{Patel_1989}
Patel, J.~K. (1989).
\newblock Prediction intervals - a review.
\newblock {\em Communications in Statistics - Theory and Methods},
  {\bf{18}}:2393--2465.

\bibitem[Schmee and Nelson, 1979]{Schmee_1979}
Schmee, J. and Nelson, W. (1979).
\newblock Predicting from early failures the last failure time of a (log)
  normal sample.
\newblock {\em IEEE Transections on Reliability}, {\bf{R-28}}:22--28.

\end{thebibliography}

\begin{table}[h]
	\centering
	\small
	\caption{Comparisons between the marginal predictors of \cite{Kaminsky_1975} and the proposed joint predictors for the  data presented in \cite{Schmee_1979}.}
	\begin{tabular}{lllll}\toprule
		\multirow{3}{*}{Marginal predictors} & \multirow{3}{*}{Joint predictors} & \multirow{3}{*}{D-efficiency} & \multirow{3}{*}{Trace-efficiency} & Overall \\
		&&&& efficiency \\
		&&&&  gain\\\midrule
		
		$\hat{X}_{(6)} = 148.73$ & 	$\tilde{X}_{(6)} = 148.58$ &  \multirow{2}{*}{0.9737} & \multirow{2}{*}{0.9937} & \multirow{2}{*}{0.0200}\\
		$\hat{X}_{(7)} = 158.99$ & 	$\tilde{X}_{(7)} = 158.86$ &&&\\
		&&&&\\
		$\hat{X}_{(6)} = 148.73$ & 	$\tilde{X}_{(6)} = 148.58$ &  \multirow{2}{*}{0.9758} & \multirow{2}{*}{0.9954} & \multirow{2}{*}{0.0196}\\
		$\hat{X}_{(8)} = 170.36$ & 	$\tilde{X}_{(8)} = 170.25$ &&&\\
		
		&&&&\\
		$\hat{X}_{(6)} = 148.73$ & 	$\tilde{X}_{(6)} = 148.58$ &  \multirow{2}{*}{0.9777} & \multirow{2}{*}{0.9968} & \multirow{2}{*}{0.0190}\\
		$\hat{X}_{(9)} = 184.37$ & 	$\tilde{X}_{(9)} = 184.29$ &&&\\
		
		&&&&\\
		$\hat{X}_{(6)} = 148.73$ & 	$\tilde{X}_{(6)} = 148.58$ &  \multirow{2}{*}{0.9798} & \multirow{2}{*}{0.9980} & \multirow{2}{*}{0.0182}\\
		$\hat{X}_{(10)} = 206.19$ & 	$\tilde{X}_{(10)} = 206.12$ &&&\\

		&&&&\\
		$\hat{X}_{(7)} = 158.99$ & 	$\tilde{X}_{(7)} = 158.86$ &  \multirow{2}{*}{0.9785} & \multirow{2}{*}{0.9966} & \multirow{2}{*}{0.0181}\\
		$\hat{X}_{(8)} = 170.36$ & 	$\tilde{X}_{(8)} = 170.25$ &&&\\
		
		&&&&\\
		$\hat{X}_{(7)} = 158.99$ & 	$\tilde{X}_{(7)} = 158.86$ &  \multirow{2}{*}{0.9806} & \multirow{2}{*}{0.9976} & \multirow{2}{*}{0.0170}\\
		$\hat{X}_{(9)} = 184.37$ & 	$\tilde{X}_{(9)} = 184.29$ &&&\\	 
		
		&&&&\\
		$\hat{X}_{(7)} = 158.99$ & 	$\tilde{X}_{(7)} = 158.86$ &  \multirow{2}{*}{0.9829} & \multirow{2}{*}{0.9985} & \multirow{2}{*}{0.0157}\\
		$\hat{X}_{(10)} = 206.19$ & 	$\tilde{X}_{(10)} = 206.12$ &&&\\

		&&&&\\
		$\hat{X}_{(8)} = 170.35$ & 	$\tilde{X}_{(8)} = 170.25$ &  \multirow{2}{*}{0.9830} & \multirow{2}{*}{0.9983} & \multirow{2}{*}{0.0152}\\
		$\hat{X}_{(9)} = 184.37$ & 	$\tilde{X}_{(9)} = 184.29$ &&&\\	 	 	 
		
		&&&&\\
		$\hat{X}_{(8)} = 170.35$ & 	$\tilde{X}_{(8)} = 170.25$ &  \multirow{2}{*}{0.9854} & \multirow{2}{*}{0.9989} & \multirow{2}{*}{0.0135}\\
		$\hat{X}_{(10)} = 206.18$ & 	$\tilde{X}_{(10)} = 206.12$ &&&\\	 	
		
		&&&&\\
		$\hat{X}_{(9)} = 184.37$ & 	$\tilde{X}_{(9)} = 184.29$ &  \multirow{2}{*}{0.9877} & \multirow{2}{*}{0.9992} & \multirow{2}{*}{0.0116}\\
		$\hat{X}_{(10)} = 206.18$ & 	$\tilde{X}_{(10)} = 206.12$ &&&\\\bottomrule
	\end{tabular}
\end{table}

\begin{landscape}
	\begin{table}[h]
		\centering
		\caption{Coefficients $\boldsymbol{a}$ and $\boldsymbol{b}$ for some selected joint predictors for the normal parent distribution.}
		\begin{tabular}{lllllllllllll}\toprule
			\multirow{2}{*}{$n$} & \multirow{2}{*}{$r|s, t$} && $a_1$ & $a_2$ & $a_3$ & $a_4$ & $a_5$ & $a_6$ & $a_7$ &
			$a_8$ & $a_9$ & $a_{10}$ \\
			& & & $b_1$ & $b_2$ & $b_3$ & $b_4$ & $b_5$ & $b_6$ & $b_7$ & $b_8$ & $b_9$ & $b_{10}$ \\\midrule 
			\multirow{6}{*}{10} & \multirow{2}{*}{5| 6, 7} && -0.1843 & -0.0322 & 0.0382 & 0.0932 & 1.0852 & & &&&\\
			&              && -0.3088 & -0.0952 & 0.0037 & 0.0812 & 1.3191 &&&&&\\
			& \multirow{2}{*}{5| 6, 10} && -0.1843 & -0.0322 & 0.0382 & 0.0932 & 1.0852 &&&&& \\
			&               && -0.8809 & -0.3849 & -0.1547 & 0.0264 & 2.3941 &&&&&\\
			& \multirow{2}{*}{5| 9, 10} && -0.6165 & -0.251 &-0.0815 & 0.0517 & 1.8973 &&&&&\\
			&               && -0.8809 & -0.3849 & -0.1547 & 0.0264 & 2.3941 &&&&&\\\midrule
			\multirow{6}{*}{15} & \multirow{2}{*}{10| 11, 12} && -0.1215 & -0.0459 &-0.0127 &0.0124 &0.0338 &0.0530 &0.0709 &0.0882 &0.1051 &0.8168 \\    
			&        && -0.1696 &-0.0754& -0.0341& -0.0027 &0.0239 &0.0478 &0.0702 &0.0917 &0.1130& 0.9351 \\
			& \multirow{2}{*}{10| 11, 15} && -0.1215 &-0.0459 &-0.0127& 0.0124 &0.0338 &0.0530 &0.0709& 0.0882 &0.1051 &0.8168\\
			&                  && -0.4160 &-0.2266 &-0.1434& -0.0803& -0.0267& 0.0215 &0.0666& 0.1102& 0.1531& 1.5415\\
			&  \multirow{2}{*}{10| 14, 15} && -0.2982 &-0.1543& -0.0912& -0.0432& -0.0025& 0.0341& 0.0683& 0.1014& 0.1339& 1.2517\\
			&                    && -0.4160 &-0.2266& -0.1434& -0.0803& -0.0267& 0.0215& 0.0666& 0.1102& 0.1531& 1.5415\\\midrule
			\multirow{12}{*}{20} & \multirow{2}{*}{5| 6, 7} && -0.1206 & -0.0361& 0.0000 & 0.0250 & 1.1324 &  & & &&\\  
			&             && -0.2030 & -0.0846 & -0.0351 & 0.0013 & 1.3214 &&&&&\\
			&  \multirow{2}{*}{5| 6, 20} && -0.1206 & -0.0361 & 0.0000 & 0.0250 & 1.1324    &&&&&\\
			&                 && -1.5471 &-0.8758 &-0.594 &-0.3862 &4.4032 &&&&&\\
			& \multirow{2}{*}{5| 19, 20} && -1.2802 &-0.7187 &-0.4830& -0.3093& 3.7912 &&&&&\\
			&                  && -1.5471 &-0.8758 &-0.5940& -0.3862& 4.4032 &&&&&\\
			& \multirow{2}{*}{10| 11, 12} && -0.0841 &-0.0310 &-0.0086& 0.0080& 0.0218& 0.0338& 0.0448& 0.0550& 0.0646& 0.8958\\
			&                 && -0.1168 &-0.0518& -0.0243& -0.0040& 0.0129& 0.0277& 0.0411& 0.0536& 0.0654& 0.9962\\
			& \multirow{2}{*}{10| 11, 20} && -0.0841 &-0.0310& -0.0086& 0.0080& 0.0218& 0.0338& 0.0448& 0.0550& 0.0646& 0.8958\\
			&                  && -0.5557 &-0.3308 &-0.2358 &-0.1652 &-0.1067 &-0.0554& -0.0087& 0.0349 &0.0763 &2.347 \\
			&   \multirow{2}{*}{10| 19, 20} && -0.4356 &-0.2545& -0.1779& -0.1211& -0.0740& -0.0327& 0.0049& 0.0400& 0.0733& 1.9774\\
			&                     && -0.5557 &-0.3308 &-0.2358& -0.1652& -0.1067& -0.0554 &-0.0087& 0.0349& 0.0763& 2.3470\\\bottomrule
			
		\end{tabular}
		
	\end{table}
\end{landscape}

\end{document}